\documentclass[11pt,a4]{article}
\usepackage{graphicx,amssymb,amsmath,amsthm,textcomp}
\usepackage{fullpage}
\usepackage{algorithm}
\usepackage{algorithmic}
\usepackage{xcolor}
\usepackage{mathtools}

\DeclarePairedDelimiter{\ceil}{\lceil}{\rceil}
\usepackage{lineno,hyperref}
\usepackage{authblk}

\newtheorem{theorem}{Theorem}[section]

\title{On $d$-distance $m$-tuple ($\ell, r$)-domination in graphs}
\author[1]{Sangram K. Jena}
\author[2]{Ramesh K. Jallu}
\author[1]{Gautam K. Das\thanks{Corresponding author}}
\affil[1]{Department of Mathematics, Indian Institute of Technology Guwahati \authorcr \{\tt sangram, gkd\}@iitg.ac.in}
\affil[2]{Institute of Computer Science, The Czech Academy of Sciences \authorcr \tt jallu@cs.cas.cz}
\date{}

\begin{document}
\maketitle

\vspace*{-1cm}

\begin{abstract}
In this article, we study the $d$-distance $m$-tuple ($\ell, r$)-domination problem. Given a simple undirected graph $G=(V, E)$, 
and positive integers $d, m, \ell$ and $r$, a subset $V' \subseteq V$ is said to be a $d$-distance $m$-tuple ($\ell, r$)-dominating 
set if it satisfies the following conditions: (i) each vertex $v \in V$ is $d$-distance dominated by at least $m$ vertices in $V'$, and 
(ii) each $r$ size subset $U$ of $V$ is $d$-distance dominated by at least $\ell$ vertices in $V'$. Here, a vertex $v$ is $d$-distance 
dominated by another vertex $u$ means the shortest path distance between $u$ and $v$ is at most $d$ in $G$. A set $U$ is $d$-distance dominated by a set of $\ell$ vertices means size of the union of the $d$-distance neighborhood of all vertices of $U$ in $V'$ is at least $\ell$. 
The objective of the $d$-distance $m$-tuple ($\ell, r$)-domination problem is to find a minimum size subset $V' \subseteq V$ 
satisfying the above two conditions.

We prove that the problem of deciding whether a graph $G$ has (i) a 1-distance $m$-tuple ($\ell, r$)-dominating set for each fixed value of $m, \ell$, and $r$, and (ii) a $d$-distance $m$-tuple  ($\ell, 2$)-dominating set for each fixed value of $d (> 1), m$, and $\ell$ of cardinality at most $k$ (here $k$ is a positive integer) are NP-complete. We also prove that for any $\varepsilon>0$, the 1-distance $m$-tuple $(\ell, r)$-domination problem and the $d$-distance $m$-tuple $(\ell,2)$-domination problem cannot be approximated within a factor of $(\frac{1}{2}- \varepsilon)\ln |V|$ and $(\frac{1}{4}- \varepsilon)\ln |V|$, respectively, unless $P = NP$.
\end{abstract}
\vspace{-0.50cm}
\section{Introduction}
\vspace{-0.20cm}
Given a simple undirected graph $G=(V,E)$, $\delta_G(v_i,v_j)$ denotes the length of a shortest path between the 
vertices $v_i$ and $v_j$ in $G$. For an integer $d >0$, the $d$-distance neighborhood of a vertex $v_i \in V$ is denoted by 
$N_G^d[v_i]$ and is defined as  $N_G^d[v_i] = \{v_j\in V \mid \delta_G(v_i,v_j) \leq d\}$. 
A \emph{$d$-distance $m$-tuple $(\ell,r)$-dominating set} (($d, m, \ell, r$) set for short) of $G$ is a subset $V'\subseteq V$ such that (i) for every $v_i \in V$, $|N^d_G[v_i]\cap V'| \geq m$, and (ii) $|(\cup_{u \in U} N^d_G[u]) \cap V'| \geq \ell$ for every $r$ size subset $U$ of $V$, where $d, m, r$, and $\ell$ are positive integers. 
If $m \geq \ell$, then the second condition in the definition of ($d, m, \ell, r$) set is redundant. In the case of 
$m = \ell (=k$, say), the ($d, m, \ell, r$) set is known as $k$-tuple dominating set in the literature. Note that, if $m = \ell$ then the value of $r>1$ is irrelevant. Therefore, we assume $r=1$ in case of $m=\ell$. From now onwards, 
we assume that $m \leq \ell$. If $d=1, m = 2, \ell = 3, r=2$ then $(d, m, \ell, r$) set is known as a liar's dominating 
set in the literature. The objective of the $d$-distance $m$-tuple $(\ell, r)$-domination problem is to find a minimum size  $d$-distance $m$-tuple $(\ell, r)$ dominating set in a given graph $G$, and we call this problem as the \emph{minimum ($d, m, \ell, r$) dominating set}  problem. In Figure \ref{fig:def}, the set of vertices $\{e,f,i\}$ form a 3-distance 2-tuple $(3, 4)$-dominating set for the graph.

\begin{figure}[!h]\vspace*{-0.25cm} 
\centering
\includegraphics[scale=0.6]{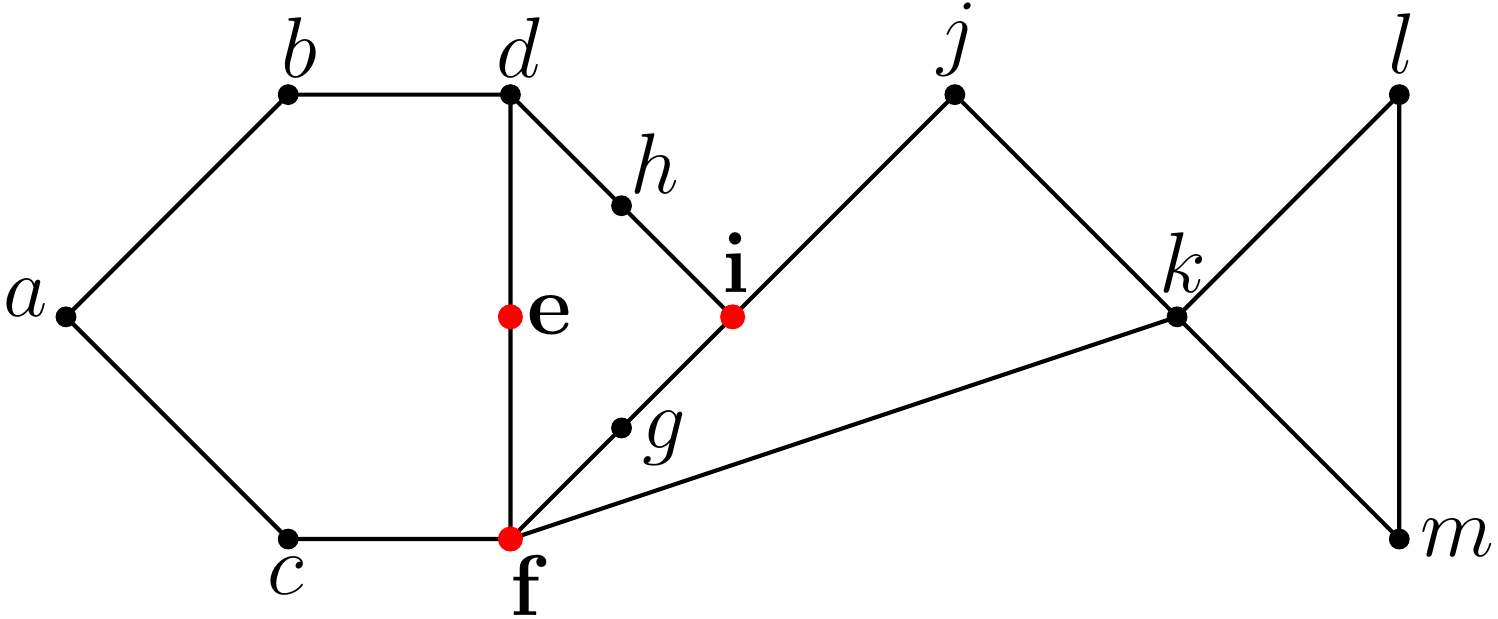}
\caption{The set $\{e,f,i\}$ is the 3-distance 2-tuple $(3, 4)$-dominating set.}
\label{fig:def}\vspace*{-0.25cm}
\end{figure}

Our interest in the problem arises from its important application such as fault tolerance in wireless/sensor networks. One specific real-time application is as follows for $r=2$. 
Suppose that in a graph $G=(V,E)$ each vertex is a possible location for an intruder such as a thief, a saboteur, a fire or some possible fault. Assume also that there is exactly $\min (\ell-m, \ceil{\ell/2}-1)$ intruders in the system represented by $G$. A protection device placed at a vertex $v$ is assumed to be able to (i) detect the intruder at any vertex in its $d$-distance neighborhood $N^d_G[v]$, and (ii) report the vertex $u \in N^d_G[v]$ at which the intruder is located. 
We are interested in deploying protection devices at a minimum number of 
vertices so that the intruder can be detected and identified correctly. 
This can be solved by finding a minimum cardinality $m$-tuple dominating set, say $D$, of $G$ and deploying protection devices at all the vertices of $D$. 
If any one protection device can fail to detect the intruder, then to correctly detect and identify the intruder one needs to place the protection devices at all the vertices of a minimum cardinality $2m$-tuple dominating set of $G$. Now it may so happen that all the protection devices detect the intruder location correctly but while reporting some of these protection devices can misreport or lie (either deliberately or through a transmission error) about the intruder location.
Assume that at most $\min (\ell-m, \ceil{\ell/2}-1)$ protection devices in the $d$-distance neighborhood of an intruder location can lie. 
Under these circumstances, to protect the network we have to install the protection devices at all the vertices of a minimum  $d$-distance $m$-tuple $(\ell, r)$ dominating set.

\vspace{-0.30cm}
\section{Related work}
\vspace{-0.20cm}
The domination problem is one of the most studied problem in the literature for its wide range of applications. Finding a minimum dominating set (MDS) in general graphs is  known to be NP-hard \cite{garey2002}. Raz and Safra \cite{raz1997sub} proved that there does not exist any approximation algorithm better than $O(\log n)$-factor unless P=NP.  The concepts of dominations and its variations are widely studied and can be seen in \cite{haynes208,haynes209}.

One of  the  variations  of  domination is the $k$-tuple domination problem and was introduced by Harary and Haynes \cite{harary}. When $k=1$, it is the usual domination problem.  For $k=2$, it is called double domination \cite{harary}. 
The same paper discusses exact values of  the  double domination numbers for  some  special graphs and various bounds  of  the  double  and  the $k$-tuple  domination numbers  in  terms  of  other  parameters. 
The hardness results and bounds for the $k$-tuple domination number for various sub-classes of graphs can be found in \cite{liao2003k,rautenbach2007new}. 

In 2009, Slater \cite{slater} first introduced 1-distance 2-tuple (3,2) domination problem known as the liar's dominating set (LDS) problem in the literature. The author proved that the problem is NP-hard for general graphs and proposed various bounds for trees, a subclass of trees, and graphs. The problem is also studied for different sub-classes of graphs and proved to be NP-hard  for  bipartite graphs  \cite{roden}, split graphs and chordal graphs \cite{panda2013liar}, doubly chordal graphs \cite{panda2015hardness},  whereas polynomially solvable in trees  \cite{panda2013liar}, block graphs \cite{panda2015hardness}, proper interval graphs \cite{paul2013}.  
Panda et al. \cite{panda2015hardness} studied the approximability of the problem and gave an $O(\ln \Delta)$-factor approximation algorithm, where $\Delta$ is the degree of the given graph. 
Alimadadi et al. \cite{alimadadi} provided the characterization of graphs and trees for which the LDS cardinality is $|V|$ and $|V|-1$, respectively. 
\vspace{-0.40cm}

\subsection{Our contribution} \label{contribution}
We prove that the problem of deciding whether a graph $G$ has a 1-distance $m$-tuple ($\ell, r$)-dominating set for each fixed value of $m, \ell$, and $r$ of cardinality at most $k$ is NP-complete (see Subsection \ref{hardness1}).
Next, we prove that the problem of deciding whether a graph $G$ has a $d$-distance $m$-tuple  ($\ell, 2$)-dominating set for each fixed value of $d (> 1), m$, and $\ell$ of cardinality at most $k$ is NP-complete (see Subsection \ref{hardness2}).
We also prove that for any $\varepsilon>0$, the 1-distance $m$-tuple $(\ell, r)$-domination problem and the $d$-distance $m$-tuple $(\ell,2)$-domination problem cannot be approximated within a factor of $(\frac{1}{2}- \varepsilon)\ln |V|$ and $(\frac{1}{4}- \varepsilon)\ln |V|$, respectively, unless $P = NP$ (see Section \ref{inapproximability}).
\section{Hardness Results} 
\subsection{Hardness of the 1-distance $m$-tuple $(\ell, r)$-domination problem}\label{hardness1}
In this section, we show that the decision version of the 1-distance $m$-tuple $(\ell, r)$-domination problem in graphs is NP-complete by reducing the {\it dominating set} (DS) problem to it, which is known to be NP-complete \cite{garey2002}.

The definition of the decision version of both the problems are as follows:
\begin{description}
\item [Decision version of 1-distance $m$-tuple $(\ell, r)$-domination problem:] 
\item[Instance:] A simple undirected graph $G = (V,E)$ with at least $\ell$ vertices and three positive integers $m$, $r$, and $k (\leq |V|)$, where $m \leq \ell$.\vspace{-0.2cm}
\item[Question:] Does $G$ has a 1-distance $m$-tuple $(\ell, r)$-dominating set of size at most $k$?
\end{description}
 
\begin{description}
\item [Decision version of the DS problem:]

\item [Instance:] A simple undirected graph $G=(V,E)$  and a positive integer $k$.
\vspace{-0.2cm}
\item [Question:] Does there exist a dominating set $D$ of $G$ such that $|D| \leq k$?
\end{description}

\begin{theorem}\label{thm:main1}
 The decision version of the \textsf{1-distance $m$-tuple $(\ell, r)$-domination} problem is NP-complete.
\end{theorem}

\begin{proof}
 For any given set $L\subseteq V$ and a positive integer $k$, we can verify whether $L$ is a 1-distance $m$-tuple 
$(\ell, r)$-dominating set of size at most $k$ or not in polynomial time by checking both the conditions of 
1-distance $m$-tuple $(\ell, r)$-dominating set. Therefore, 1-distance $m$-tuple $(\ell, r)$-domination problem is in NP.

Now, we prove the hardness of the 1-distance $m$-tuple $(\ell, r)$-domination problem by reducing the decision version of the \textsc{DS} problem, which is known to be NP-complete \cite{garey2002}, to it.
Let $<G=(V,E),k>$ be an instance of the dominating set problem, where $G=(V,E)$ is an undirected graph with vertex set $V=\{v_1,v_2, \ldots ,v_n\}$ and $k$ is an integer. 
We construct an instance $<G'=(V',E'),m,\ell, r>$ of the decision version of 1-distance $m$-tuple $(\ell, r)$-domination problem as follows: 

\begin{minipage}{.5\textwidth}
\begin{center}
\begin{align*}
V'&=V^1\cup V^2\cup V^3, \text{where}\\ 
V^1&=\{v_1^1,v_2^1, \ldots ,v_n^1\},\\
V^2&=\{v_1^2,v_2^2, \ldots v_{\ell-1}^2\},\\
V^3&=\{v_1^3,v_2^3, \ldots, v_r^3\}\\
\end{align*}
\end{center}
\end{minipage}%
\begin{minipage}{.5\textwidth}
\begin{center}
\begin{align*}
 E'&=E^1\cup E^2\cup E^3\cup E^4, \text{where}\\ 
 E^1&=\{(v_i^1,v_j^1)\mid (v_i,v_j) \in E\},\\
 E^2&=\{(v_i^2,v_j^2)\mid 1\leq i<j\leq \ell-1\},\\
 E^3&=\{(v_i^1,v_j^2)\mid 1\leq i \leq n, 1\leq j\leq \ell-1\},\\
 E^4&=\{(v_i^2,v_j^3)\mid 1\leq i \leq \ell-1, 1 \leq j \leq r\}  
\end{align*}
\end{center}
\end{minipage}%

Observe that, $G'=(V',E')$ can be constructed in polynomial time and $|V'|=n+\ell+r-1$, where $n=|V|$ and $\ell, r <n$. 
An illustration for the construction of $G'$ from $G$ is shown in Figure \ref{hardnessfig1}(a).

\begin{figure}[!t]
\centering
\includegraphics[scale=0.75]{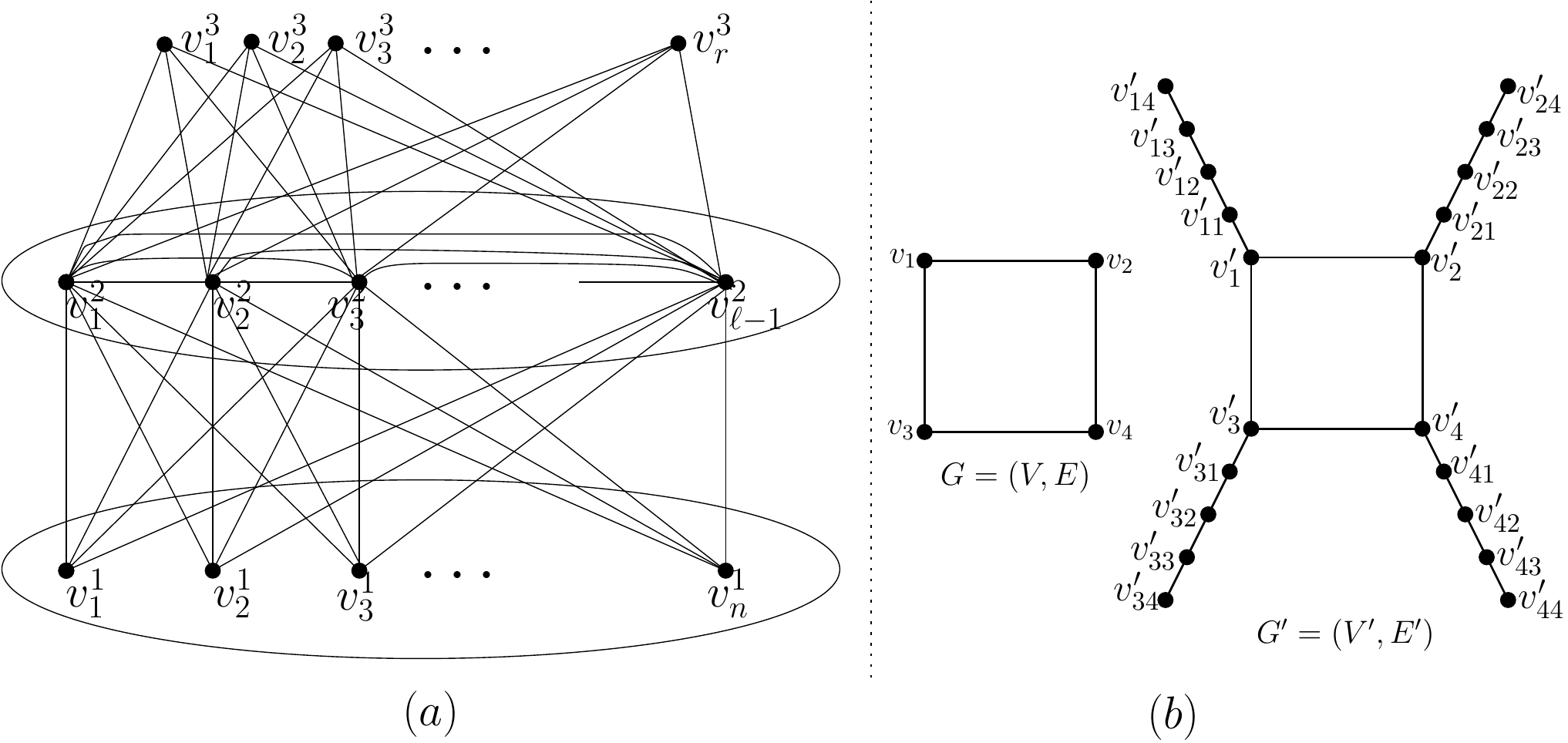}
\caption{(a) A graph $G'=(V',E')$ constructed for an instance of the 1-distance $m$-tuple $(\ell,r)$-domination problem, and (b) a graph $G'=(V',E')$ constructed for an instance of the $d$-distance $m$-tuple $(\ell,2)$- domination problem.}\label{hardnessfig1} \vspace{-0.40cm}
\end{figure}

{\bf Claim 1:} {\it $G$ has a dominating set of size at most $k$ if and only if $G'$ has a 1-distance $m$-tuple $(\ell, r)$-dominating set of size at most $k+\ell$}.

{\it \bf Proof:} Let $D$ be a dominating set of $G$ and $|D|\leq k$. 
Let $L=\{v_i^1 \mid v_i\in D\}\cup V^2\cup \{v_1^3\}$. 
Now, we show that $L$ is a 1-distance $m$-tuple $(\ell, r)$-dominating set in $G'$.

{\bf (i)} Observe that for each $v \in V'$, $|N^1_{G'}[v] \cap L| \geq m$ as $m \leq \ell$ (value of $r=1$ in case of $m = \ell$) 
and each $v \in V'$ is dominated by $\ell-1$ vertices in $V^2$.

{\bf (ii)} Let $U = \{u_1, u_2, \ldots, u_r\} \subseteq V'$ be an arbitrary subset of size $r$. 

{\bf Case 1:} Let $U \cap V^2 \neq \emptyset$ and $v_i^2 \in U \cap V^2$. From the construction of $G'$, 
$N_G^1[v_i^2] \cap L \supseteq V^2 \cup \{v_1^3\}$, which implies $|N_G^1[v_i^2] \cap L| \geq \ell$. Therefore, 
$|(\cup_{u \in U} N_G^1[u]) \cap L| \geq \ell$. 

{\bf Case 2:}  Let $U \cap V^1 \neq \emptyset$ and $v_i^1 \in U \cap V^1$. From the constructions of $G'$ and $L$, 
$N_G^1[v_i^1] \cap L \supseteq V^2 \cup \{v_j^1\}$, where $v_j \in D$ is a dominator of $v_i$ in $G$. Therefore, $|N_G^1[v_i^1] \cap L| \geq \ell$, which leads to $|(\cup_{u \in U} N_G^1[u]) \cap L| \geq \ell$. 

{\bf Case 3:} Let $U = V^3$. Again, from the constructions of $G$ and $L$, $(\cup_{u \in U} N_G^1[u]) \cap L \supseteq V^2 \cup \{v_1^3\}$. Therefore, 
in this case also $|(\cup_{u \in U} N_G^1[u]) \cap L| \geq \ell$. 

Thus $L$ is a 1-distance $m$-tuple $(\ell,r)$-dominating set in $G'$ and  $|L| \leq k+\ell$.

Conversely, let $L$ be a 1-distance $m$-tuple $(\ell,r)$-dominating set for $G'$ of size at most $k+\ell$. 
From the definition of the 1-distance $m$-tuple $(\ell,r)$-dominating set and as $|V^3| = r$,  
$|(\cup_{v \in V^3} N^1_{G'}[v] ) \cap L|\geq \ell$. 
Therefore, there must be at least $\ell$ vertices from $V^2 \cup V^3$ in $L$ (see Figure \ref{hardnessfig1}(a)). 
Let $D=\{v_i\in V\mid v_i^1\in L \setminus (V^2 \cup V^3)\}$. 
If $D$ is a dominating set of $G$, then we are done as $|D| \leq k$. 
Suppose $D$ is not a dominating set in $G$. 
Since $|V^2|$ is $\ell-1$, the 1-distance neighborhood of every subset of $V^1$ with cardinality greater than or equal to $r$ will have a non-empty intersection with $D$ (due to the second condition of 1-distance $m$-tuple $(\ell,r)$-domination). This implies, for any subset $U^1$ of $V$, $D \cap (\cup_{u \in U^1}N^1_G[u])=\emptyset$ if and only if $|U^1|\leq r-1$. Note that such a set $U^1 (\neq \emptyset)$ exists 
based on our assumption that $D$ is not a dominating set of $G$. Let $|U^1|=s$. Now, we will show that $|(V^2 \cup V^3) \cap L| \geq \ell+s$. 

Let $U^2 (\subseteq V^2)$ and $U^3 (\subseteq V^3)$ be the maximum size subsets such that $U^2 \cap L = \emptyset$ and 
 $U^3 \cap L = \emptyset$, respectively. Let $s' = |U^2|$ and $s'' = |U^3|$. Let $U_{13} = U^1 \cup U^3$. Since 
$(\cup_{u \in U_{13}}N^1_{G'}[u])\cap L = V^2\setminus U^2$, i.e., $|(\cup_{u \in U_{13}}N^1_{G'}[u])\cap L| = \ell-1-s' < \ell$, 
$|U_{13}| = s+s'' < r$. Add $r-s-s''$ vertices from $V^3\setminus U^3$ to the vertex set $U_{13}$. Now, by the definition of $L$ 
the size of the set $|(\cup_{u \in U_{13}}N_{G'}^1[u])\cap L|$ must be at least $\ell$. Therefore, $r-s-s'' + \ell-1-s' \geq \ell$, 
which implies $r-s'' \geq s+s'+1$.

Since $r-s'' \geq s+s'+1$, $|D|\leq k-s$. Let $D_1 = D \cup U^1$. So, every vertex in $V$ is dominated by at least one vertex in $D_1$ whose size is at most $k$. Therefore, we conclude, the decision version of 1-distance $m$-tuple $(\ell,r)$-domination problem is NP-complete. 
\end{proof}

\subsection{Hardness of the $d$-distance $m$-tuple {$(\ell,2)$-domination} problem}\label{hardness2}
In this section, we show that the decision version of $d$-distance $m$-tuple $(\ell,2)$-domination problem is NP-complete.
 For fixed constant $d \geq 2$, the decision version of the problem is defined as follows.

\begin{description}
\item [Instance:] An undirected connected graph $G = (V,E)$ with $|V|\geq \ell$ and three positive integers $m$, $d$, and $k (\leq |V|)$, where $m \leq \ell$.
\vspace{-0.1in}
\item [Question:] Does $G$ has a $d$-distance $m$-tuple $(\ell,2)$-dominating set of size at most $k$?
\end{description}

We prove that decision version of $d$-distance $m$-tuple $(\ell,2)$-domination problem ($d \geq 2$) is NP-complete by reducing the decision 
version of the 1-distance $m$-tuple $(\ell,2)$-domination problem to it in polynomial time. Note that 1-distance $m$-tuple 
$(\ell,2)$-domination problem is NP-complete (see Section \ref{hardness1}). Recall, the decision version of 1-distance $m$-tuple $(\ell,2)$-domination problem: 

\begin{description}
 \item [Instance:] An undirected connected graph $G = (V,E)$ with $|V|\geq \ell$ and two positive integer $m, k \leq |V|$, where $m \leq \ell$.
\vspace{-0.1in}
 \item [Question:] Does $G$ has a 1-distance $m$-tuple $(\ell,2)$-dominating set of size at most $k$?
\end{description}

\begin{theorem}\label{thm:decissiontheorem}
The decision version of the $d$-distance $m$-tuple $(\ell,2)$-domination problem is NP-complete.
\end{theorem}

\begin{proof}
The decision version of the $d$-distance $m$-tuple $(\ell,2)$-domination problem is in NP as for a given certificate (a subset of $V$) we can verify whether it is satisfying both the conditions of the $d$-distance $m$-tuple $(\ell,2)$-dominating set or not in polynomial time.

We now describe a polynomial time reduction from an arbitrary instance 
of the decision version of 1-distance $m$-tuple $(\ell,2)$-domination problem to an instance of the 
decision version of the $d$-distance $m$-tuple $(\ell,2)$-domination problem.

Let $G=(V=\{v_1,v_2, \ldots ,v_n\},E)$ be an arbitrary instance of the decision version of 1-distance $m$-tuple 
$(\ell,2)$-domination problem. We construct an instance, a graph $G'=(V',E')$, of the decision version of the 
$d$-distance $m$-tuple $(\ell,2)$-domination problem as follows:\vspace{-0.60cm}

\begin{align*}
V'&=\{v_i'\mid v_i \in V\} \cup (\bigcup_{v_i \in V}\{v_{i1}',v_{i2}',\ldots,v_{id-1}'\})   (see\ Figure\ \ref{hardnessfig1}(b)\ for\ an\ example)\\
E'&=\{(v_i',v_j') \mid (v_i,v_j)\in E\}\cup 
(\bigcup_{v_i \in V}\{(v_i',v_{i1}'),(v_{i1}',v_{i2}'),\ldots,(v_{id-2}',v_{id-1}')\})
\end{align*}\vspace{-0.40cm}

\noindent
{\bf Claim 2:} {\it $G$ has a 1-distance $m$-tuple $(\ell,2)$-dominating set of cardinality at most $k$ if and only if 
$G'$ has a $d$-distance $m$-tuple $(\ell,2)$-dominating set of cardinality at most $k$}.\\
{\bf Necessity:} Let $L$ be a 1-distance $m$-tuple $(\ell,2)$-dominating set of $G$ such that $|L|\leq k$. 
Let $L'=\{v_i' \in V' \mid v_i \in L\}$. We can argue that $L'$ is a $d$-distance $m$-tuple $(\ell,2)$-dominating set in $G'$ and $|L'| \leq k$.
Since $|L'| = |L|$ and $|L|\leq k$, so $|L'| \leq k$. As each vertex $v \in V$ satisfies 
1-distance $m$-tuple $(\ell,2)$-domination properties and each vertex in $G'$ is at most $d-1$ distance 
away from a vertex in $L'$, $L'$ suffices to ensure $d$-distance $m$-tuple $(\ell,2)$-dominating set in graph $G'$ for $d \geq 2$.

{\bf Sufficiency:} Let $L'$ be a $d$-distance $m$-tuple $(\ell,2)$-dominating set in $G'$ such that $|L'|\leq k$. We shall show that, 
by updating (i.e., removing or replacing) some of the vertices in $L'$, at most $k$ vertices from 
$\{v_1',v_2',\ldots,v_n'\}$ can be chosen such that the set of corresponding vertices in $V$ is an  
1-distance $m$-tuple $(\ell,2)$-dominating set in $G$. Let $L''=L'$. For each vertex $v_{ij}' \in V'$, 
$(1 \leq j \leq d-1$ and $1 \leq i \leq n)$ we do the following: if $v_{ij}' \in L''$, then replace it with its 
associated vertex $v_i'$ if $v_i'$ is not already in $L''$, otherwise, replace it with any vertex in 
$N^1_{G'}[v_i'] \cap \{v_1',v_2',\ldots,v_n'\}$ which is not in $L''$. If all the vertices of 
$N^1_{G'}[v_i'] \cap \{v_1',v_2',\ldots,v_n'\}$ are in $L''$ 
(i.e., $(N^1_{G'}[v_i']\cap \{v_1',v_2',\ldots,v_n'\}) \subseteq L'' )$, then remove $v_{ij}'$ from $L''$. 
Therefore, $|L''| \leq k$. Let $L=\{v_i \in V \mid v_i' \in L''\}$. Now, we prove that $L$ is an 
1-distance $m$-tuple $(\ell,2)$-dominating set in $G$ such that $|L|\leq k$. 

Since $|L''|\leq k$, then $|L| \leq k$. We first prove the first condition (i.e., for every $v \in V$, 
$|N^1_G[v]\cap L| \geq m$) of 1-distance $m$-tuple $(\ell,2)$-dominating set. Consider a vertex  $v_i' \in V'$, for some 
$1 \leq i \leq n$, let $s$ be the number of vertices in $L' \cap \{v_{i1}',v_{i2}',\ldots,v_{id-1}'\}$.

{\bf Case 1. $s=0$.}  
Since $L'$ is $d$-distance $m$-tuple $(\ell,2)$-dominating set, there must exist at least $m$ vertices, say $\{v_1'',v_2'', \ldots, v_m''\}$ in $\{v_1',v_2', \ldots,v_n'\}\cap L'$ such that $\{v_1'',v_2'', \ldots , v_m''\} \subseteq N^d_{G'}[v_{i,d-1}']$, otherwise, $L'$ is not a feasible solution as $v_{id-1}'$ does not have $m$ distance-$d$ $(m,\ell)$-dominators. 
Therefore, $|N^1_{G'}[v_i'] \cap (\{v_1',v_2', \ldots,v_n'\} \cap L'')|\geq m$.\\
{\bf Case 2. $s \geq 1$.} 
Let $v_{ij_1}',v_{ij_2}', \ldots, v_{ij_t}'\in L'$, for some $1\leq j_1,j_2, \ldots, j_t\leq d-1$. 
By our construction of $L''$ each vertex in $\{v_{ij_1}',v_{ij_2}', \ldots, v_{ij_t}'\}$ is replaced by one of the vertices in $N^1_{G'}[v_i'] \cap \{v_1',v_2', \ldots,v_n'\}$. 
Therefore, in this case also $|N^1_{G'}[v_i'] \cap (\{v_1',v_2', \ldots,v_n'\}\cap L'')| \geq m$. Thus, by our construction of $L$ from $L''$, $|N^1_G[v_i] \cap L|\geq m$ is true.

Now we prove the second condition of 1-distance $m$-tuple $(\ell,2)$-dominating set (i.e., for every pair of distinct vertices $u,v \in V$, $|(N^1_G[u]\cup N^1_G[v])\cap L| \geq \ell$).

Let $v_i$ and $v_j$ be two distinct vertices in $G$. Consider the vertices $v_{id-1}'$ and $v_{jd-1}'$ in $G'$.  
As $L'$ is a $d$-distance $m$-tuple $(\ell,2)$-dominating set of $G'$, it satisfies the second property of $d$-distance $m$-tuple $(\ell,2)$-domination in $G'$. 
Thus there exist at least $\ell$ dominators dominating $v_{id-1}'$ and $v_{jd-1}'$ in $L'$, i.e., $|(N^d_{G'}[v_{id-1}']\cup N^d_{G'}[v_{jd-1}'])\cap L'|\geq \ell$.
These dominators are either from $N^1_G[v_i'] \cup N^1_G[v_j']$ or from 
$\{v_{i1}',v_{i2}',\ldots,v_{id-1}'\}$ and/or from $\{v_{j1}',v_{j2}' \ldots,v_{jd-1}'\}$.
As per our construction of $L''$ from $L'$, we are replacing each dominator in $\{v_{i1}',v_{i2}',\ldots,v_{id-1}'\} \cup \{v_{j1}',v_{j2}' \ldots,v_{jd-1}'\} $ (if any) by a vertex in $(N^1_{G'}[v_i'] \cup N^1_{G'}[v_j']) \cap \{v_1',v_2',\ldots,v_n'\}$.

Since $G$ is connected and $|V|\geq \ell$, so is $G'$. Therefore, $L''$ contains  at least $\ell$ vertices from $(N^1_{G'}[v_i'] \cup N^1_{G'}[v_j']) \cap \{v_1',v_2',\ldots,v_n'\}$, i.e., $|(N^1_{G'}[v_i'] \cup N^1_{G'}[v_j']) \cap \{v_1',v_2',\ldots,v_n'\} \cap L''| \geq \ell$. 
Therefore, according to the construction of $L$ from $L''$, $|(N^1_G[v_i] \cup N^1_G[v_j]) \cap L| \geq \ell$. 
Thus, $L$ is a 1-distance $m$-tuple $(\ell,2)$-dominating set of the graph $G$ having cardinality at most $k$.

 Therefore, the decision version of $d$-distance $m$-tuple $(\ell,2)$-domination problem is NP-complete.
\end{proof}
\vspace{-0.70cm}
\section{Inapproximability results}\label{inapproximability}

\subsection{Inapproximability of the 1-distance $m$-tuple $(\ell,r)$-domination problem} \label{inapproximability1}
In this section, we prove that the 1-distance $m$-tuple $(\ell,r)$-domination problem cannot be approximated 
within a factor of $(\frac{1}{2}- \varepsilon) \ln(|V|)$ for any $\varepsilon >0$, unless P = NP.
We argue the claim by showing that if 1-distance $m$-tuple $(\ell,r)$-domination problem can be approximated within a factor of $(\frac{1}{2}- \varepsilon) \ln(|V|)$ for any $\varepsilon >0$ in a graph $G'$, then the domination problem can be approximated within a factor of $(1- \varepsilon) \ln(|V|)$ for any $\varepsilon >0$. 

\begin{theorem} \cite{dinur} \label{thm:dinur}
For every $\varepsilon >0$, it is NP-hard to approximate set cover problem within a factor of $(1- \varepsilon) \ln n$, where $n$ is the size of the instance. The reduction runs in $n^{O(1/\varepsilon)}$ time. 
\end{theorem}

\begin{theorem} \label{thm:inapprox}
Minimum domination problem cannot be approximated within a factor of $(1- \varepsilon) \ln(|V|)$ for any $\varepsilon >0$, unless P = NP.
\end{theorem}

\begin{proof}
The result follows from (i) the relation between set cover problem and dominating set problem, (ii) Theorem 
\ref{thm:dinur}, and (iii) the inapproximability result in \cite{chleb}.
\end{proof}

\begin{theorem}\label{thm:own}
Minimum 1-distance $m$-tuple $(\ell,r)$-domination problem cannot be approximated within a factor of 
$(\frac{1}{2}- \varepsilon) \ln(|V|)$ for any $\varepsilon >0$, unless P = NP.
\end{theorem}

\begin{proof}
Let $G$ be a simple graph.
Consider the construction of the graph $G'$ for any given graph $G$ as discussed in Section \ref{hardness1}.
As per our construction, we proved that each instance of domination problem can be 
reducible to an instance of 1-distance $m$-tuple $(\ell,r)$-domination problem in polynomial-time . \\
Let $D^*$ and $L^*$ be the optimal DS and 1-distance $m$-tuple $(\ell,r)$-dominating set in $G$ and $G'$, 
with cardinalities $\gamma_{ds}(G)$  and $\gamma_{m \ell r}(G')$, respectively. Now we can argue the following 
claim: $\gamma_{m \ell r}(G')=\gamma_{ds}(G)+\ell$. The inequality $\gamma_{m \ell r}(G')\leq \gamma_{ds}(G)+\ell$ 
is trivial as per our construction in Section \ref{hardness1}. On the other hand, 
$\gamma_{m \ell r}(G')\geq \gamma_{ds}(G)+\ell$ follows from the sufficiency proof of Claim 1 in Section \ref{hardness1}. 
So given a dominating set $D$ of $G$, one can find a 1-distance $m$-tuple $(\ell,r)$-dominating set $L$ of $G'$ such that 
$|L|=|D|+\ell$. Now, $\frac{|L|}{|L^*|}=\frac{|D|+\ell}{|D^*|+\ell} \geq \frac{1}{2} \frac{|D|}{|D^*|}$. Suppose 
there exists a polynomial time algorithm that approximates 1-distance $m$-tuple $(\ell,r)$-domination problem within 
a factor of $(\frac{1}{2}- \varepsilon) \ln N$ for graphs with $N$ vertices. As per our construction of the graph $G'$ from 
$G$ (see Figure \ref{hardnessfig1}(a)), $G'$ contains, $N=n+\ell+r - 1 \leq 3n$ for  $n \geq 2$ vertices, where $n$ is the 
total number of vertices in $G$, $\ell < n$, and $r < n$. Therefore,\vspace{-0.50cm}

\begin{align*}
\frac{|D|}{|D^*|}\leq  (1- 2\varepsilon) \ln N \leq (1- 2\varepsilon) \ln n(1+ \frac{\ln 4}{\ln n}).
\end{align*}\vspace{-0.50cm}

For sufficiently large $n$, the term $(1+ \frac{\ln 4}{\ln n})$ can be bounded by $1+\frac{\varepsilon}{5}$, 
where $\varepsilon \geq \frac{5 \ln 4}{\ln n}$. Now we have\vspace{-0.70cm}

\begin{align*}
(1- 2\varepsilon) \ln n(1+ \frac{\ln 4}{\ln n}) \leq (1- 2\varepsilon)[\ln n + \ln (1+ \frac{\varepsilon}{5})] \leq 
(1- 2\varepsilon)[\ln n +  \frac{\varepsilon}{5} \ln n] \leq (1- \varepsilon')\ln n, 
\end{align*}\vspace{-0.50cm}

where $\varepsilon' < \frac{9}{5} \varepsilon+ \frac{2}{5}\varepsilon^2$. Therefore, for an arbitrary graph, we can 
approximate the domination problem by a factor of $(1-\varepsilon')\ln n$, which leads to a contradiction to 
Theorem \ref{thm:inapprox}. Thus, the minimum 1-distance $m$-tuple $(\ell,r)$-domination problem cannot 
be approximated within a factor of $(\frac{1}{2}- \varepsilon) \ln(|V|)$ for any $\varepsilon >0$, unless P = NP.
\end{proof}

\subsection{Inapproximability of the $d$-distance $m$-tuple $(\ell,2)$-domination problem}\label{inapproximability2}

In this section, we give a lower bound on the approximation ratio of any approximation algorithm for the 
$d$-distance $m$-tuple $(\ell,2)$-domination problem by providing an approximation preserving reduction from the 
1-distance $m$-tuple $(\ell,r)$-domination problem for $r=2$.

\begin{theorem}\label{thm:inaprox}
Given a simple undirected graph $G=(V,E)$, the $d$-distance $m$-tuple $(\ell,2)$-domination problem cannot be approximated within a 
factor of $(\frac{1}{4}- \varepsilon)\ln |V|$, for any fixed constant $d\geq 2$ and $\varepsilon>0$, unless P = NP.
\end{theorem}

\begin{proof}
Let $G=(V,E)$ be an arbitrary instance of the 1-distance $m$-tuple $(\ell,2)$-domination problem with $n$ vertices.
Given $G=(V,E)$, we construct a graph $G'=(V',E')$, an instance of the $d$-distance $m$-tuple $(\ell,2)$-domination problem as described in Section \ref{hardness2}. Let $L^*$ and $L_d^*$ be the optimal 1-distance $m$-tuple $(\ell,2)$-dominating set and 
$d$-distance $m$-tuple $(\ell,2)$-dominating set in $G$ and $G'$, with cardinalities $\gamma_{m \ell}(G)$  and $\gamma_{m \ell}^d(G')$, respectively.
Now we can argue the following claim: $\gamma_{m \ell}^d(G')=\gamma_{m \ell}(G)$.
The inequality $\gamma_{m \ell}^d(G')\leq \gamma_{m \ell}(G)$ is trivial as every 1-distance $m$-tuple $(\ell,2)$-dominating set 
of $G$ is a $d$-distance $m$-tuple $(\ell,2)$-dominating set in $G'$. On the other hand, $\gamma_{m \ell}^d(G')=|L_d^*|\geq |L|$ follows from the 
sufficiency proof of Claim 2 in Section \ref{hardness2}. \\
Given any 1-distance $m$-tuple $(\ell,2)$-dominating set $L$ of $G$, one can find a $d$-distance $m$-tuple $(\ell,2)$-dominating set $L_d$ 
of $G'$ with $|L_d|=|L|$. Suppose there exist a polynomial time algorithm to approximate $d$-distance $m$-tuple $(\ell,2)$-domination problem 
within a factor of $(\frac{1}{4}-\varepsilon)\ln |V'|$, where $|V'| = n + n(d-1) \leq n^2$ (see Section \ref{hardness2}).
Now $\frac{|L|}{|L^*|}=\frac{|L_d|}{|L_d^*|}\leq (\frac{1}{4}-\varepsilon)\ln n^2 = (\frac{1}{2}-2\varepsilon)\ln n\leq (\frac{1}{2}-\varepsilon')\ln n$, where $\varepsilon' \leq 2\varepsilon$. 
Therefore, the result follows from Theorem \ref{thm:own}.
\end{proof}

 \section{Conclusion}\label{sec:conclusion}
In this article, we studied $d$-distance $m$-tuple ($\ell, r$)-domination problem. We provided a common NP-completeness 
proof of the 1-distance $m$-tuple ($\ell, r$) domination problem for each fixed value of $m, \ell$, and 
$r$. We also presented a common NP-completeness proof of the $d$-distance $m$-tuple  ($\ell, 2$) domination 
problem for each fixed value of $d (> 1), m$, and $\ell$. We have showed that the first problem is not approximated 
within a factor of $(\frac{1}{2}- \varepsilon)\ln |V|$ for each fixed value of $m, \ell$, and $r$, unless P = NP 
and the second problem is not approximated within a factor of $(\frac{1}{4}- \varepsilon)\ln |V|$ for each fixed value of  $d (> 1), m$, and $\ell$, unless P = NP, where $V$ is the vertex set of the input graph. The reduction in the NP-completeness/inapproximability proofs are very powerful as these are common reductions for completely different kind of dominations. 


\vspace{-0.15in}
\bibliographystyle{plain}

\end{document}